\documentclass[12pt]{amsart}

\usepackage{amsmath,amsthm,amssymb,amsfonts,verbatim}
\usepackage{graphicx}
\usepackage[colorlinks,
            linkcolor=blue,
            anchorcolor=blue,
            citecolor=blue
            ]{hyperref}
\usepackage[hmargin=1.2in,vmargin=1.2in]{geometry}

\title[Asymptotic construction of LRC]{Constructive asymptotic bounds of locally repairable codes via function fields}
\author{Liming Ma}\address{School of Mathematical Sciences, Yangzhou University, Yangzhou, China 225002}\email{lmma@yzu.edu.cn}
\author{Chaoping Xing} \address{School of Electronics, Information and Electric Engineering, Shanghai Jiao Tong University, 
China 200240}\email{xingcp@sjtu.edu.cn}
\date{}

\newtheorem{lemma}{Lemma}[section]
\newtheorem{theorem}[lemma]{Theorem}
\newtheorem{cor}[lemma]{Corollary}

\newtheorem{prop}[lemma]{Proposition}

\theoremstyle{remark}
\newtheorem{rmk}{Remark}

\renewcommand{\epsilon}{\varepsilon}
\renewcommand{\le}{\leqslant}
\renewcommand{\ge}{\geqslant}

\newcommand{\vnote}[1]{}

\def\F{\mathbb{F}}

\def \mL {\mathcal{L}}

\def\Pin{{P_{\infty}}}

\def \Xi {{X^{[i]}}}

\newcommand{\Ga}{\alpha}
\newcommand{\Gb}{\beta}

\newcommand{\Gl}{\lambda}

\def \ba {{\bf a}}
\def \bc {{\bf c}}
\def \bi {{\bf 1}}
\def \bx {{\bf x}}

\def\bh {{\bf h}}

\def \bo {{\bf 0}}
\def \g{\mathfrak{g}}

\def\LRC {{\rm locally repairable code\ }}
\def\LRCs {{\rm locally repairable codes\ }}

\begin{document}

\maketitle

\begin{abstract} 
Locally repairable codes have been investigated extensively in recent years due to practical applications in distributed and cloud storage systems. However, there are few asymptotical constructions of locally repairable codes in the literature. In this paper, we provide an explicit asymptotic construction of locally repairable codes over arbitrary finite fields from local expansions of functions at a rational place. This construction gives a Tsfasman-Vladut-Zink type bound for locally repairable codes. Its main advantage is that there are no constraints on both locality and alphabet size. Furthermore, we show that the Gilbert-Varshamov type bound on locally repairable codes over non-prime finite fields can be exceeded for sufficiently large alphabet size. 
\end{abstract}

\section{Introduction}
Because of practical applications in distributed and cloud storage systems, locally repairable codes have been studied by many researchers \cite{FY14, GHSY12, GXY19, J19, LXY19, PD14, PKLK12, SRKV13, TB14, TPD16, XY18}. 
A code is said with locality $r$ if  every erasure of a given codeword can be recovered by accessing at most $r$ other symbols of this codeword.  
Unlike in the classical coding case, only a few papers study the asymptotical behavior of locally repairable codes  \cite{BTV17, CM15, LMX19, TBF16}. 
The main purpose of this paper is to present a new explicit construction of asymptotically good locally repairable codes via function fields. 

\subsection{Locally repairable codes and some bounds}\label{subsec:2.1}

Let $q$ be a prime power and let $\F_q$ be the finite field with $q$ elements. Let $C\subseteq \F_q^n$ be a $q$-ary block code of length $n$. For each $\Ga\in\F_q$ and $i\in \{1,2,\cdots, n\}$, define $C(i,\Ga):=\{\bc=(c_1,\dots,c_n)\in C\; : \; c_i=\Ga\}$. For a subset $I\subseteq \{1,2,\cdots, n\}\setminus \{i\}$, we denote by $C_{I}(i,\Ga)$ the projection of $C(i,\Ga)$ on $I$. Then $C$ is called a locally repairable code with locality $r$ if, for every $i\in \{1,2,\cdots, n\}$, there exists a subset
$I_i\subseteq \{1,2,\cdots, n\}\setminus \{i\}$ with $|I_i|\le r$ such that  $C_{I_i}(i,\Ga)$ and $C_{I_i}(i,\Gb)$ are disjoint for any $\Ga\neq \Gb\in \F_q$.

A linear \LRC over $\F_q$ of length $n$, dimension $k$, minimum distance $d$ and locality $r$ is denoted to be a $q$-ary $[n,k,d]$-linear code with locality $r$.
It is proved in \cite{GHSY12} that an $[n,k,d]$-linear code with locality $r$ satisfies the Singleton type bound
\begin{equation}\label{Singletonbound}
d\le n-k-\left\lceil \frac{k}{r} \right\rceil+2.
\end{equation}
 A code achieving the bound (\ref{Singletonbound}) is usually called an optimal locally repairable code. 
There are many different techniques to construct optimal locally repairable codes. 
One powerful method among them is to construct optimal locally repairable codes from automorphism groups of function fields \cite{BHHMV17, JMX17, LMX19, TB14}. 

In this paper, we mainly focus on the asymptotical behavior of locally repairable codes. 
The locality $r$ is fixed, but the dimension and minimum distance are proportional to the length $n$.
Let $R_q(r,\delta)$ denote the asymptotic bound on the rate of $q$-ary locally repairable codes with locality $r$ and relative minimum distance $\delta$, i.e., $$R_q(r,\delta)=\limsup_{n\rightarrow \infty} \frac{\log_qM_q(n,\lfloor\delta n\rfloor,r)}{n} ,$$
where $M_q(n,d,r)$ is the maximum size of  \LRCs of length $n$, minimum distance $d$ and locality $r$.

There are various asymptotically upper bounds on locally repairable codes. 
The Singleton type bound \eqref{Singletonbound} gives
\begin{equation}\label{S_bound}
R_q(r,\delta)\le  \frac{r}{r+1}(1-\delta) \mbox{ for } 0\le \delta\le 1.
\end{equation}
The asymptotic Plotkin type bound is given by
\begin{equation}\label{P_bound}
R_q(r,\delta)\le  \frac{r}{r+1}\left(1-\frac{q}{q-1}\cdot \delta\right) \text{ for } 0\le \delta\le 1-q^{-1}. 
\end{equation}
The following bound \eqref{LP_bound} is derived from the linear programming bound given in \cite{Aa79}
\begin{equation}\label{LP_bound}
R_q(r,\delta)\le \min_{0\le \tau \le \frac{1}{r+1}} \left\{\tau r+(1-\tau (r+1))f_q\left(\frac{\delta}{1-\tau(r+1)}\right)\right\},
\end{equation}
where $f_q(x):=H_q\left(\frac{1}{q}\big[ q-1-x(q-2)-2\sqrt{(q-1)x(1-x)}\big]\right)$ and $H_q(x)$ is the $q$-ary entropy function defined by
$$H_q(x):=x\log_q(q-1)-x\log_q(x)-(1-x)\log_q(1-x).$$

For $0\le \delta \le 1-q^{-1}$, the asymptotic Gilbert-Varshamov bound of locally repairable codes is given by 
\begin{equation}\label{asymp_GV_bound}
R_q(r,\delta)\ge 1-\min_{0<s\le 1} \left\{\frac{1}{r+1} \log_q\Big([1+(q-1)s]^{r+1}+(q-1)(1-s)^{r+1}\Big)-\delta \log_qs\right\}
\end{equation}
in \cite{TBF16}.

\subsection{Known results}
Although there are several asymptotically upper bounds and  the asymptotic Gilbert-Varshamov bound on locally repairable codes,  there is little work on asymptotical lower bounds  that are constructive.

For the classical codes, it is well known that  the Tsfasman-Vladut-Zink bound can improve upon the Gilbert-Varshamov bound in a certain interval for any square prime power $q\ge 49$ (see \cite{TVZ82} or \cite[Theorem 8.4.7]{St09}). In order to construct asymptotically good locally repairable codes, the algebraic geometry codes should be a good candidate. 

Using asymptotically optimal Garcia-Stichtenoth tower of function fields,
Barg {\it et al.} \cite{BTV17} gave a construction of  asymptotically good $q$-ary locally repairable codes with locality $r$ whose rate $R$ and relative distance $\delta$ satisfy
\begin{equation}\label{construction_l}
R \ge\frac{r}{r+1}\Big{(}1-\delta-\frac{3}{\sqrt{q}+1}\Big{)},\quad r=\sqrt{q}-1,
\end{equation}
and 
\begin{equation}\label{construction_l+1}
R \ge \frac{r}{r+1}\Big{(}1-\delta-\frac{\sqrt{q}+r}{q-1}\Big{)},\quad (r+1)|(\sqrt{q}+1).
\end{equation}
It was further shown in \cite{BTV17} that for some values $r$ and $q$, the bound \eqref{construction_l+1} exceeds the asymptotic Gilbert-Varshamov bound \eqref{asymp_GV_bound} on locally repairable codes.

Li {\it et al.} \cite{LMX17} generalized the idea given in \cite{BTV17} by considering more subgroups of automorphism groups of function fields in the Garcia-Stichtenoth tower. 
This construction allows more flexibility of locality. If $r+1=up^v$ with $u|(p^v-1, \sqrt{q}-1)$, then there exists a family of explicit $q$-ary linear locally repairable codes with locality $r$ whose rate $R$ and relative distance $\delta$ satisfy
\begin{equation}\label{eq:lmx17}R\ge \frac{r}{1+r}\Big{(}1-\delta-\frac{\sqrt{q}+r-1}{q-\sqrt{q}}\Big{)}.\end{equation}

There are two shortcomings for the above bounds obtained from automorphism groups of function fields in the Garcia-Stichtenoth tower. 
The first shortcoming is that the alphabet size $q$ must be a square of prime power, and the second one is the restriction on the locality $r$, i.e., $r+1$ must be a divisor of $\sqrt{q}+1$ or $\sqrt{q}(\sqrt{q}-1)$.
In this paper, we will overcome these shortcomings by using local expansions of functions at a rational place.

\subsection{Our results and comparison}
In this paper, we provide a new asymptotic construction of locally repairable codes via function fields. The underlying idea is based on the technique of local expansions to construct algebraic geometry codes which was initiated in \cite{XNL99}.  The difficulty is how to endow algebraic geometry codes with the additional structure of locality. 
Such locally repairable codes are obtained from parity-check matrices whose columns are formed by coefficients of local expansions of carefully chosen functions at a rational place.
Our main results of this paper are summarized below. 

\begin{theorem}\label{thm: 1.1}
Let $q$ be a prime power and let $A(q)$ be the Ihara's constant. 
Then there exists a family of $q$-ary linear locally repairable codes with locality $r$ whose rate $R$ and relative distance 
$\delta$ satisfy
\begin{equation}\label{eq:main1}
R\ge \frac{r}{r+1}-\frac{1}{A(q)} \frac{r}{r+1}-\delta.
\end{equation}
\end{theorem}

If $q=p^n$ is a prime power with $n\ge 2$, then there exist explicit towers of function fields which obtain a good lower bound of $A(q)$ \cite{BBGS15, GS95, GS96}. Hence, this construction  is explicit for locally repairable codes over non-prime finite fields. In particular, we have the following results.

\begin{cor}\label{cor: 1.2}
\begin{itemize}
\item[(i)] If $q$ is a square, then there exists an explicit family of $q$-ary linear locally repairable codes with locality $r$ whose rate $R$ and relative distance $\delta$ satisfy
\begin{equation}\label{eq:cor1.2}
R\ge \frac{r}{r+1}-\frac{1}{\sqrt{q}-1} \frac{r}{r+1}-\delta.
\end{equation}
\item[(ii)]  If $q$ is an odd power of prime, i.e., $q=p^{2m+1}$ with $m\ge 1$,  then there exists an explicit family of $q$-ary linear locally repairable codes with locality $r$ whose rate $R$ and relative distance $\delta$ satisfy
\begin{equation}\label{eq:cor1.2_2}
R\ge \frac{r}{r+1}-\frac{1}{2}\left(\frac{1}{p^m-1}+\frac{1}{p^{m+1}-1}\right) \frac{r}{r+1}-\delta.
\end{equation}
\end{itemize}
\end{cor}

\begin{proof}
If $q$ is a square, then the Garcia-Stichtenoth tower is the well-known explicit tower of function fields such that $A(q)=\sqrt{q}-1$ from \cite{GS95}. Thus, the item (i) follows immediately from Theorem \ref{thm: 1.1}. 
If $q$ is an odd power of prime, then there is an explicit tower of function fields such that
$$A(p^{2m+1})\ge \frac{2(p^{m+1}-1)}{p+1+\zeta} \text{ with } \zeta=\frac{p-1}{p^m-1}$$
from \cite{BBGS15}. Thus, the item (ii) follows immediately from Theorem \ref{thm: 1.1}. 
\end{proof}

\begin{figure}
\centering
\includegraphics[width=3.5in]{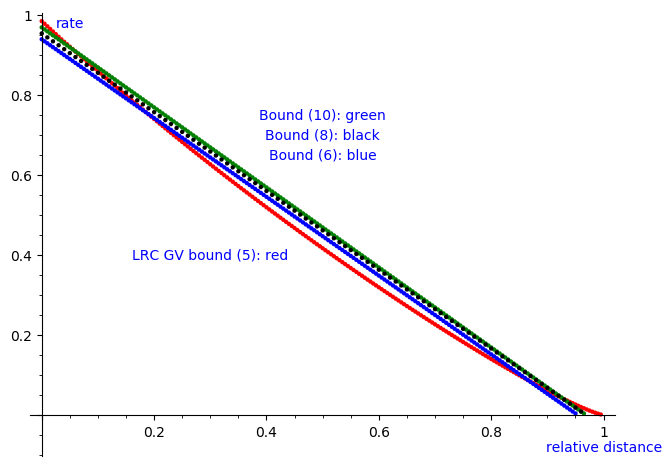}
\caption{$r=63, q=2^{12}$}\label{fig:1}
\end{figure}

\begin{figure}
\centering
\includegraphics[width=3.5in]{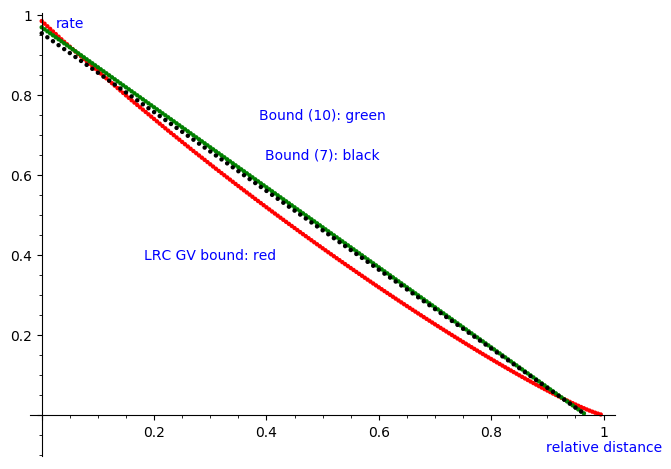}
\caption{$r=64, q=2^{12}$}\label{fig:2}
\end{figure}

The main advantage of this construction is to allow arbitrary locality. 
The bound (\ref{eq:cor1.2}) given in Corollary \ref{cor: 1.2} is better than the bound (\ref{construction_l+1}) given in \cite{BTV17} if and only if $$\delta< \frac{(r-1)r}{q-1}.$$
The bound (\ref{eq:cor1.2}) given in Corollary \ref{cor: 1.2} is better than the bound (\ref{eq:lmx17}) given in \cite{LMX17} if and only if $$\delta< \frac{(r-1)r}{q-\sqrt{q}}.$$
Hence, the bound  (\ref{eq:cor1.2})  is better than the bound  (\ref{construction_l+1}) or (\ref{eq:lmx17}) if $r>\sqrt{q}$. 

The figures  \ref{fig:1} and \ref{fig:2} show that the bound (\ref{eq:cor1.2}) given in Corollary \ref{cor: 1.2}  can exceed the asymptotic Gilbert-Varshamov bound (\ref{asymp_GV_bound}) on locally repairable codes, the bound (\ref{construction_l}) or (\ref{construction_l+1}) given in \cite{BTV17}  and the bound (\ref{eq:lmx17}) given in \cite{LMX17} for $r=63,q=2^{12}$ and $r=64,q=2^{12}$, respectively. The figure \ref{fig:3} shows that the bound (\ref{eq:cor1.2_2}) given in Corollary \ref{cor: 1.2}  can exceed the asymptotic Gilbert-Varshamov bound (\ref{asymp_GV_bound}) of locally repairable codes for $r=64,q=2^{13}$ as well. 

\begin{figure}\centering
\includegraphics[width=3.5in]{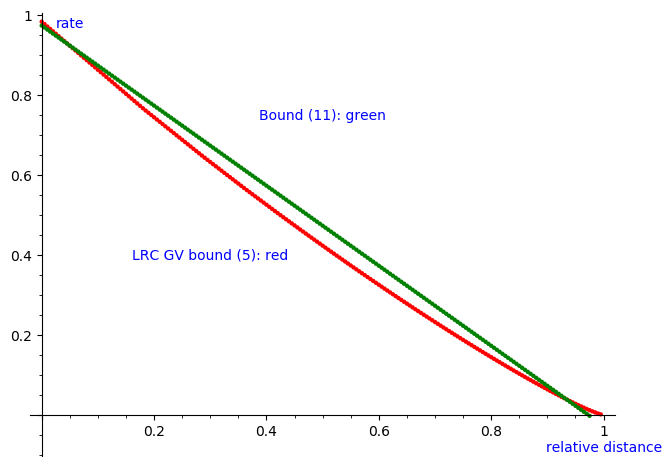}
\caption{$r=64, q=2^{13}$}\label{fig:3}
\end{figure}

\begin{figure}\centering
\includegraphics[width=3.5in]{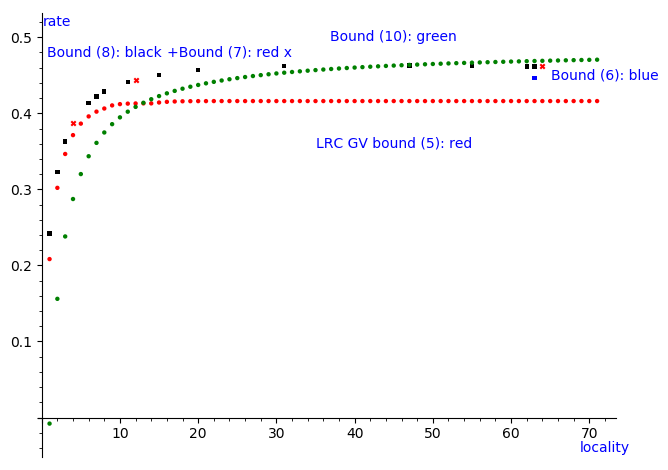}
\caption{$\delta=0.5, q=2^{12}$}\label{fig:4}
\end{figure}

There are no constraints on locality $r$ in Theorem \ref{thm: 1.1} compared with the bounds given in \cite{BTV17, LMX17}. 
The figure \ref{fig:4} shows that the bound (\ref{eq:cor1.2}) given in Corollary \ref{cor: 1.2} can exceed the asymptotic Gilbert-Varshamov bound of locally repairable codes given for  many localities for $\delta=0.5, q=2^{12}$.

Furthermore, we can show that our bound \eqref{eq:cor1.2} given in Corollary \ref{cor: 1.2} always exceeds the asymptotic Gilbert-Varshamov bound on locally repairable codes for some range of locality $r$ when $q$ is sufficiently large.

\begin{prop}\label{prop: 1.3}
Let $\F_q$ be a non-prime finite field. 
If locality $r$ lies in the range $[c\log_{2} q, +\infty)$ for any constant $c>1$, then the bound \eqref{eq:cor1.2} given in Corollary \ref{cor: 1.2} exceeds the asymptotic Gilbert-Varshamov bound \eqref{asymp_GV_bound}  of locally repairable codes over $\F_q$ for all sufficiently large $q$.
\end{prop}

Explicit asymptotically good towers of function fields over $\F_q$ are of great interest for coding theory, since they can be applied to construct asymptotically good families of linear codes over $\F_q$. But there doesn't exist explicit towers of function fields over a prime finite field in the literature. 
Using places with high degrees of function fields, the method based on local expansions of functions at a rational place can be generalized to construct asymptotic locally repairable codes over prime finite fields as well. 

\begin{theorem}\label{thm: 1.4}
Let $q$ be a prime and let $r$ be an integer. Let $b$ be an integer which is defined as follows:
$$b=\begin{cases} r,   & \text{ if } q \mbox{ is odd and } r \mbox{ is even};\\ r+1, & \text{ if } r \mbox{ is odd};\\ r+2, &  \text{ if } q=2 \mbox{ and } r \mbox{ is even}.  \end{cases}$$
Let $e$ be an even divisor of $b$.  Then there exists a family of $q$-ary linear locally repairable codes with locality $r$ whose rate $R$ and relative distance 
$\delta$ satisfy
\begin{equation}\label{eq:thm1.4}
R\ge \frac{r}{r+1} -  \frac{b}{r+1} \frac{1}{q^{\frac{e}{2}}-1} - e\delta.
\end{equation}
\end{theorem}

In particular, the rate $R$ of locally repairable codes over $\F_2$ with locality $r=11$ is lower bounded by
$$R\ge \begin{cases} 227/252-12\delta, & \text{ if }  0\le \delta \le 4/189;\\ 65/84-6\delta, & \text{ if } 4/189 \le \delta \le 2/21; \\ 7/12-4\delta,  & \text{ if } 2/21\le \delta \le 7/48.\end{cases}$$
The figure \ref{fig:5} gives a comparison of the bound \eqref{eq:thm1.4}  and the asymptotical Gilbert-Varshamov bound \eqref{asymp_GV_bound} for $r=11, q=2$. 
The main advantage of this explicit construction is that there are no constraints on locality over a prime finite field. 
Unfortunately, it seems that the bound \eqref{eq:thm1.4} given in Theorem \ref{thm: 1.4} can't exceed the asymptotical Gilbert-Varshamov bound on locally repairable codes. 

\begin{figure}\centering
\includegraphics[width=3.5in]{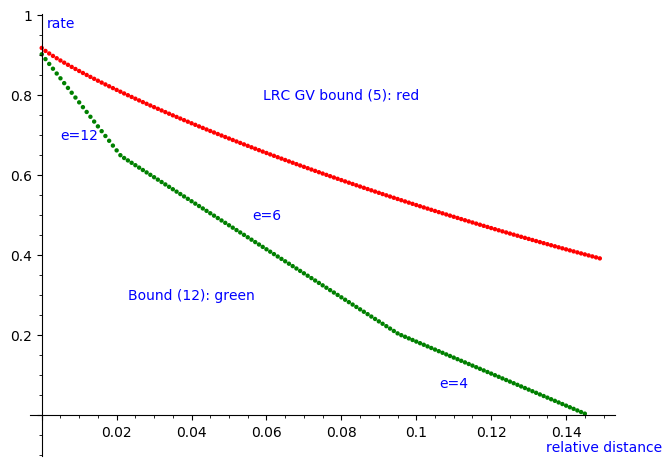}
\caption{$r=11, q=2$}\label{fig:5}
\end{figure}

\subsection{Organization}
This paper is organized as follows. In Section \ref{sec:2}, we introduce some preliminaries on function fields including Riemann-Roch space, local expansion, Ihara's constant $A(q)$ and Garcia-Stichtenoth tower. In Section \ref{sec: 3}, we present our asymptotical construction of locally repairable codes from local expansions of functions at a rational place over non-prime finite fields. This method via function fields is generalized to construct asymptotical locally repairable codes over prime finite fields in Section \ref{sec: 4}.

\section{Preliminaries}\label{sec:2}

In this section, we present some preliminaries on algebraic function fields over finite fields, such as Riemann-Roch space, local expansion, Ihara's constant and Garcia-Stichtenoth tower. The reader may refer to \cite{NX01, St09} for more details.

\subsection{Riemann-Roch space}
Let $F/\F_q$ be an algebraic function field over the full constant field $\F_q$.
Let $\mathbb{P}_F$ denote the set of places of $F$ and let  $\g(F)$ denote the genus of $F$.
The principal divisor of $z\in F^*$ is defined by $$(z)=\sum_{P\in \mathbb{P}_F} \nu_P(z)P,$$ 
where $\nu_P$ is the normalized discrete valuation with respect to the place $P$.
Let $G$ be a divisor of $F$. 
The Riemann-Roch space associated to $G$ is defined by $$\mathcal{L}(G)=\{z\in F^*: (z)\ge -G\}\cup \{0\}.$$
It turns out to be a finite-dimensional vector space over $\F_q$ and its dimension, denoted by $\ell(G)$, is lower bounded by $\ell(G)\ge \deg(G)-\g(F)+1$. If $\deg(G)\ge 2\g(F)-1$, then $\ell(G)=\deg(G)-\g(F)+1$ from Riemann-Roch theorem \cite[Theorem 1.5.17]{St09}.

\subsection{Local expansion}
For a rational place $P$ of $F$, there exists an element $t\in F$ with $\nu_P(t)=1$, which is called a local parameter at $P$. 

For any nonzero function $f\in F$, we can find an integer $v$ such that $\nu_P(f)\ge v$. 
Hence, we have $\nu_P(f/t^v)\ge 0$. Put $a_v=(f/t^v)(P)$. Then $a_v\in \F_q$ and $\nu_P(f/t^v-a_v)\ge 1$. 
It follows that $$\nu_P \left (\frac{f-a_vt^v}{t^{v+1}}\right )\ge 0.$$ Put $a_{v+1}=(f-a_vt^v)/t^{v+1}+P$.
Then we have $a_{v+1}\in \F_q$ and \[\nu_P(f-a_vt^v-a_{v+1}t^{v+1})\ge v+2.\]
By iterating the above process recursively, we can obtain an infinite sequence $\{a_i\}_{i=v}^{\infty}$ in $\F_q$
such that $$\nu_P\left (f-\sum_{i=v}^{m}a_i t^i\right )\ge m+1$$ for all $m\ge v$. 
The formal power series $$f=\sum_{i=v}^{\infty}a_i t^i,$$ is called a \textit{local expansion} of $f$ at $P$. 

Conversely, if $\{a_i\}_{i\ge v}$ is a sequence in $\F_q,$ then the series $\sum_{i=v}^{\infty}a_it^i$ converges in the $P$-adic completion of $F$ and  $$\nu_P\left (\sum_{i=v}^{\infty}a_it^i\right )=\min\{i: a_i\neq 0\}$$ from \cite[Theorem 4.2.6]{St09}.

\subsection{Ihara's constant $A(q)$}
For an integer $\g\ge 0$, let $N_q(\g)$ denote the maximum number of rational places of all function fields over $\F_q$ of genus $\g$. 
The real number defined by $$A(q)=\limsup_{g\rightarrow \infty} \frac{N_q(\g)}{\g}$$ is called Ihara's constant.
If $q$ is a square of prime power, the lower bound $$A(q)\ge \sqrt{q}-1$$
was proved by Ihara in \cite{Ih81} and Tsfasman et al. \cite{TVZ82}, using modular curves.
By refining Ihara's method, Drinfeld and Vladut \cite{DV83}  obtained a tight upper bound $$A(q)\le \sqrt{q}-1.$$
Hence, $A(q)=\sqrt{q}-1$ if $q$ is a square of prime power. For any non-square prime power, the exact value of $A(q)$ is unknown.

For applications in coding theory, one needs algebraic function fields over $\F_q$ with large genus and many rational places, which are given in explicit equations and their rational places can be obtained explicitly as well. 
Garcia and Stichtenoth provided the first explicit tower of function fields over $\F_q$ with $q=\ell^2$ such that the Drinfeld-Vladut bound is achieved \cite{GS95}. 
For an odd power of prime $p^{2m+1}$, Bassa {\it et al.}  \cite{BBGS15} constructed an explicit tower of function fields over $\F_q$ such that the following lower bound can be obtained
$$A(p^{2m+1})\ge \frac{2(p^{m+1}-1)}{p+1+\zeta} \text{ with } \zeta=\frac{p-1}{p^m-1}.$$

\subsection{Garcia-Stichtenoth tower}\label{asymtower}
Let  $q=\ell^2$ be a square of prime power.
The Garcia-Stichtenoth tower of function fields $\mathcal{F}=(F_1,F_2,F_3,\cdots)$ is explicitly given by the rational function field $F_1=\F_q(y_1)$ and $F_m=F_{m-1}(y_m)$ with
\begin{equation}\label{optimaltower}
y_{m}^{\ell}+y_{m}=\frac{y_{m-1}^{\ell}}{y_{m-1}^{\ell -1}+1}
\end{equation}
recursively for $m\ge 2$.
The main properties of the Garcia-Stichtenoth tower are summarized in the following proposition \cite{GS96}.

\begin{prop}\label{prop: 2.1} 
\begin{itemize}
\item[(i)]  The number of rational places of $F_m$ is lower bounded by $N(F_m)\ge (q-\ell)\ell^{m-1}+\ell.$
\item[(ii)]  The genus of $F_m$ is given by
$$\g(F_m)=\begin{cases}
(\ell^{\frac{m}{2}}-1)^2, & \mbox{ if } m \equiv 0(\mbox{mod } 2),\\
(\ell^{\frac{m+1}{2}}-1)(\ell^{\frac{m-1}{2}}-1), & \mbox{ if } m \equiv 1(\mbox{mod } 2).
\end{cases}$$
\item[(iii)]  The Garcia-Stichtenoth tower  $\mathcal{F}=(F_1,F_2,F_3,\cdots)$ is asymptotically optimal, since it attains the Drinfeld-Vladut bound over $\F_q$, i.e., 
$$\lim_{m\rightarrow \infty} \frac{N(F_m)}{\g(F_m)}=\sqrt{q}-1.$$
\end{itemize}
\end{prop}

\section{The Asymptotic Tsfasman-Vladut-Zink bound for LRCs}\label{sec: 3}

Let $F/\F_q$ be a global function field of genus $\g$ and let $\Pin$ be a distinguished rational place of $F$.  
The dimension of the Riemann-Roch space $\mL((2\g-1)P_{\infty})$ is equal to $\g$ from Riemann-Roch theorem. 
Let $\{f_1,f_2,\dots,f_\g\}$ be an $\F_q$-basis of $\mL((2\g-1)\Pin)$. Then there exist $\g$ distinct pole numbers $0=n_1<n_2<\cdots <n_\g\le 2\g-1$ such that $(f_j)_{\infty}=n_jP_{\infty}$ for $1\le j\le \g$ from Weierstrass gap theorem \cite[Theorem 1.6.8]{St09}.
Let $\pi\in F$ be a local parameter of $P_{\infty}$. For $1\le j\le n$, the local expansions of $f_j$ at $\Pin$ are given in the following form 
\begin{equation}\label{eq:expansion-f}
f_j=\pi^{-2\g+1}\sum_{i=0}^{\infty}c_{i,j}\pi^i,
\end{equation}
where all coefficients $c_{i,j}\in \F_q$. For a non-negative integer $t$,  consider the $(2\g+t)\times \g$ matrix

\begin{equation}\label{eq:A}
A=\left(\begin{array}{cccc}
c_{0,1}& c_{0,2}& \cdots& c_{0,\g}\\
c_{1,1}&c_{1,2}& \cdots& c_{1,\g}\\
\vdots& \vdots& \vdots&\vdots\\
c_{2\g-1+t,1}&c_{2\g-1+t,2}& \cdots&c_{2\g-1+t,\g}
\end{array}\right).
\end{equation}

\begin{lemma}\label{lem:rank} 
The matrix $A$ defined in \eqref{eq:A} has rank $\g$.
\end{lemma}

\begin{proof} Suppose that the rank of $A$ is strictly less than $\g$.   Let $\ba_1,\ba_2,\dots,\ba_\g$ be the columns of $A$.  
Then there exist elements $\Gl_1,\Gl_2,\dots,\Gl_\g$ in $\F_q$, not all are zero, such that $\sum_{j=1}^\g\Gl_j\ba_j=\bo.$

The local expansion of $\sum_{j=1}^\g\Gl_j f_j$ has the form $\pi^{-2\g+1}\sum_{i=2\g+t}^{\infty}c_{i}\pi^i$, where $c_i=\sum_{j=1}^\g\Gl_jc_{i,j}$ for all $i\ge 2\g+t$. 
This means that $\nu_\Pin(\sum_{j=1}^\g\Gl_j f_j)\ge 1+t>0.$
On the other hand, $\sum_{j=1}^\g\Gl_j f_j$ is a function of $\mL((2\g-1)\Pin)$. This forces that $\sum_{j=1}^\g\Gl_j f_j$ is the zero function, because it is easy to see $\nu_\Pin(\sum_{j=1}^\g\Gl_j f_j)\le 0$ from the choices of $f_j$ for $1\le j\le \g$.
Hence, $\Gl_j=0$ for all $1\le j\le n$ which lead to a contradiction.
\end{proof}

\begin{rmk}\label{rmk:rank}{\rm \begin{itemize}
\item[(i)] From the proof of Lemma \ref{lem:rank}, it is easy to see that the submatrix of $A$ defined in \eqref{eq:A} consisting of the first $2\g$ rows has rank $\g$.
\item[(ii)] By using Gaussian elimination, one can efficiently find an invertible $\g\times \g$ submatrix of $A$.
\end{itemize}
}\end{rmk}

Let $r$ be a fixed positive integer.  Let $\{P_{ij}: 1\le i\le m, 1\le j\le r\}$ be a set of distinct rational places of $F$ which are different from $P_{\infty}$. 
For each $P_{ij}$, we extend $\{f_1,f_2,\dots,f_\g\}$ to an $\F_q$-basis $\{f_1,f_2,\dots,f_\g,g_{ij}\}$ of $\mL((2\g-1)\Pin+P_{ij})$. Furthermore, we have $\nu_{P_{ij}}(g_{ij})=-1$ for $1\le i\le m, 1\le j\le r$. 

Without loss of generality, we assume that the submatrix of $A$ defined in \eqref{eq:A} consisting of the first $\g$ rows has rank $\g$. Denote by $A_1$ the submatrix of $A$  consisting of the first $\g$ rows.
Assume that $g_{ij}$ has the following local expansions at $\Pin$:
\begin{equation}\label{eq:expansion-g}
g_{ij}=\pi^{-2\g+1}\sum_{p=0}^{\infty}b_{pij}\pi^p.
\end{equation}
Let $(\Ga_{1ij},\Ga_{2ij},\dots,\Ga_{\g ij})\in\F_q^\g$ be the unique solution of the system of linear equations $A_1\bx^T=(b_{0ij},b_{1ij},\dots,b_{\g-1,i,j})^T$. 
Put $f_{ij}=g_{ij}-\sum_{w=1}^\g\Ga_{wij}f_w$. Then it is easy to verify that local expansions of $f_{ij}$ are given in the form
 \begin{equation}\label{eq:expansion-fij}
f_{ij}=\pi^{-2\g+1}\sum_{p=\g}^{\infty}a_{pij}\pi^p,
\end{equation}
where $a_{pij}=b_{pij}-\sum_{w=1}^\g\Ga_{wij} c_{pw}$. Furthermore, we have the following results.

\begin{lemma}\label{lem: 3.2} 
One has the following facts:
\begin{itemize}
\item[{\rm (i)}] $f_{ij}\in\mL((2\g-1)\Pin+P_{ij})$ and $\nu_{P_{ij}}(f_{ij})=-1$ for all $1\le i \le m$ and $1\le j\le r$;
\item[{\rm (ii)}]  $\nu_{P_{uv}}(f_{ij})\ge 0$ whenever $u\neq i$ or $v\neq j$;
\item[{\rm (iii)}] $\{f_{ij}: 1\le i\le m, 1\le j\le r\}$ are linearly independent over $\F_q$.
\end{itemize}
\end{lemma}

\begin{proof} 
\begin{itemize}
 \item[{\rm (i)}]  It is easy to see that $f_{ij}=g_{ij}-\sum_{w=1}^\g\Ga_{wij}f_w\in \mL((2\g-1)\Pin+P_{ij})$ and $\nu_{P_{ij}}(g_{ij})=-1$.  From the strict triangle inequality \cite[Lemma 1.1.11]{St09}, we have $$\nu_{P_{ij}}(f_{ij})=\min \Big\{\nu_{P_{ij}}(g_{ij}), \nu_{P_{ij}}\big{(}\sum_{w=1}^\g\Ga_{wij}f_w\big{)}\Big\}=  -1.$$
\item[{\rm (ii)}] 
It follows directly from $f_{ij}=g_{ij}-\sum_{w=1}^\g\Ga_{wij}f_w\in \mL((2\g-1)\Pin+P_{ij})$. 
\item[{\rm (iii)}] 
Suppose that there exist $\Gl_{ij}\in \F_q$ for $1\le i\le m, 1\le j\le r$, not all are zero, such that $$\sum_{i=1}^m \sum_{j=1}^{r}\Gl_{ij} f_{ij}=0.$$
Assume that $\Gl_{i_0,j_0}\neq 0$ for some $1\le i_0\le m$ and $1\le j_0\le r$. Using the strict triangle inequality, we have
$$\nu_{P_{i_0,j_0}}\left(\sum_{i=1}^m \sum_{j=1}^{r}\Gl_{ij} f_{ij}\right)=-1$$
from items (1) and (2). Hence, we obtain a contradiction. 
\end{itemize}
\end{proof}

For each $1\le i\le m$, let $1\neq \alpha_i\in \F_q^*$, $f_{i,r+1}=\alpha_i f_{i1}$ and define the matrix
\begin{equation}\label{eq:D}
D_i=\left(\begin{array}{ccccc}
a_{\g, i,1}& a_{\g, i,2}& \cdots& a_{\g, i,r}& \alpha_i a_{\g, i,1}\\
a_{\g+1, i,1}& a_{\g+1, i,2}& \cdots& a_{\g+1, i,r}& \alpha_i a_{\g+1, i,1}\\
\vdots&\vdots & \vdots& \vdots&\vdots\\
a_{2\g-1+t,i,1} &a_{2\g-1+t,i,2}& \cdots&a_{2\g-1+t,i,r} &  \alpha_i a_{2\g-1+t,i,1}
\end{array}\right).
\end{equation}
The $j$-column of $D_i$ is obtained from local expansions of $f_{ij}$ for each $1\le j\le r+1$. 
Furthermore, we define the matrix
\begin{equation}\label{eq:H3} H=\left(\begin{array}{c|c|c|c}
\bi&\bo&\cdots&\bo\\ \hline
\bo&\bi&\cdots&\bo \\ \hline
\vdots&\vdots&\ddots&\vdots \\ \hline
\bo&\bo&\cdots&\bi \\ \hline
D_1&D_2&\cdots&D_m
\end{array}
\right),
\end{equation}
where $\bi$ and $\bo$ stand for the all-one vector and the zero vector of length $r+1$, respectively.

\begin{prop}\label{prop: 3.3}
Let $H$ be the matrix defined in \eqref{eq:H3}.
Let $C$ be the linear code with $H$ as a parity-check matrix. 
Then $C$ is a $q$-ary $[n,k,d]$-linear code with locality $r$ and
\[n=m(r+1),\quad k\ge n-\frac n{r+1}-\g-t,\quad d\ge t+1.\]
\end{prop}

\begin{proof}
It is obvious that the length of $C$ is $n=m(r+1)$. 
Since the number of rows of the parity-check matrix $H$ is $m+\g+t$,  the dimension of $C$ is at least $n-m-\g-t$. 
It is sufficient to prove that any $t$ columns of $H$ are linearly independent over $\F_q$. 

Choose any $t$ columns $\{\bh_{ij}\}_{1\le i\le m;j\in S_i}$ of $H$, where $S_i$ are subsets of $\{1,2,\dots,r+1\}$ satisfying $\sum_{i=1}^m|S_i|=t$. 
Define the subset $I=\{1\le i\le m:\; |S_i|\ge 2\}$. Let $\{\Gl_{ij}\}_{1\le i\le m;j\in S_i}$ be the elements of $\F_q$ such that $\sum_{i=1}^m\sum_{j\in S_i}\Gl_{ij}\bh_{ij}=\bo$, i.e., 
$$\sum_{i\in I}\sum_{j\in S_i}\Gl_{ij}\bh_{ij}=-\sum_{i\not\in I}\sum_{j\in S_i}\Gl_{ij}\bh_{ij}.$$
Firstly, we claim that $\Gl_{ij}=0$ for all $i\in \{1,2,\dots,m\}\setminus I$ and $j\in S_i$. 
Otherwise, we may assume that $\Gl_{i_0,j_0}\neq 0$ for some $i_0\notin I$ and $j_0\in S_{i_0}$. 
In this case, we must have $|S_{i_0}|=|\{j_0\}|=1$. 
Hence, the $i_0$-th position of $-\sum_{i\not\in I}\sum_{j\in S_i}\Gl_{ij}\bh_{ij}$ is $-\Gl_{i_0,j_0}$, while the $i_0$-th position of $\sum_{i\in I}\sum_{j\in S_i}\Gl_{ij}\bh_{ij}$ is $0$. 
This is a contradiction. Thus, we have $$\sum_{i\in I}\sum_{j\in S_i}\Gl_{ij}\bh_{ij}=\bo.$$
This implies that the local expansion of the function $\sum_{i\in I}\sum_{j\in S_i}\Gl_{ij}f_{ij}$ is
\[\pi^{-2\g+1}\sum_{p=2\g+t}^{\infty}\left(\sum_{i\in I}\sum_{j\in S_i}\Gl_{ij}a_{pij}\right)\pi^p.\]
Therefore, $\nu_{\Pin}(\sum_{i\in I}\sum_{j\in S_i}\Gl_{ij}f_{ij})\ge 1+t$. 
On the other hand, $\sum_{i\in I}\sum_{j\in S_i}\Gl_{ij}f_{ij}$ belongs to the Riemman-Roch space $\mL((2\g-1)\Pin+\sum_{i\in I} \sum_{j\in S_i} P_{ij})$.  
Hence, we have \[\sum_{i\in I}\sum_{j\in S_i}\Gl_{ij}f_{ij}\in \mL\Big{(}-(1+t)\Pin+\sum_{i\in I} \sum_{j\in S_i} P_{ij}\Big{)}.\]
As $\deg\left(-(1+t)\Pin+\sum_{i\in I} \sum_{j\in S_i} P_{ij}\right)= -(1+t)+\sum_{i\in I}|S_i|\le -(1+t)+t= -1$,  
we have $$\sum_{i\in I}\sum_{j\in S_i}\Gl_{ij}f_{ij}=0.$$

It remains to show that $\Gl_{ij}=0$ for all $i\in I$ and $j\in S_i$.
If there doesn't exist $u\in I$ such that $1,r+1$ are in $S_u$ at the same time, then it is easy to see that $\Gl_{ij}=0$ for all $i\in I$ and $j\in S_i$ from Lemma \ref{lem: 3.2}(iii).
Otherwise, we have
$$-\sum_{i\in I\setminus \{u\}}\sum_{j\in S_i}\Gl_{ij}f_{ij}=\sum_{j\in S_u} \Gl_{uj}f_{uj}=(\Gl_{u,1}+\alpha_u \Gl_{u,r+1})f_{u1}+\sum_{j\in S_u\setminus \{1,i+1\}}\Gl_{uj}f_{uj}.$$
{\bf Case 1:} If there exists some $j_0\in S_u\setminus \{1,i+1\}$ such that $\Gl_{u,j_0}\neq 0$, then 
we have $\nu_{P_{u,j_0}}(-\sum_{i\in I\setminus \{u\}}\sum_{j\in S_i}\Gl_{ij}f_{ij})\ge 0$ and $\nu_{P_{u,j_0}}(\sum_{j\in S_u} \Gl_{uj}f_{uj})=-1$ from Lemma \ref{lem: 3.2}. This is impossible. \\
{\bf Case 2:} Otherwise, $\Gl_{u,j}=0$ for all $j\in S_u\setminus \{1,i+1\}$. If $\Gl_{u,1}+\alpha_u \Gl_{u,r+1}\neq 0$, then we have $\nu_{P_{u,1}}(-\sum_{i\in I\setminus \{u\}}\sum_{j\in S_i}\Gl_{ij}f_{ij})\ge 0$ and $\nu_{P_{u,1}}(\sum_{j\in S_u} \Gl_{uj}f_{uj})=-1$ from Lemma \ref{lem: 3.2}. This is impossible.
Thus, we have $\Gl_{u,1}+\alpha_u \Gl_{u,r+1}=0$.
Moreover, it is easy to see that $\sum_{j\in S_u} \Gl_{u,j}=0$ from the parity-check matrix $H$. 
Note that $\alpha_u\neq 1$. 
Thus, $\Gl_{u,j}=0$ for all $j\in S_u$.
Recursively, we can show that $\Gl_{ij}=0$ for all $i\in I$ and $j\in S_i$.
\end{proof}

By employing asymptotically good towers of function fields \cite{BBGS15, GS96}, we can obtain an explicit construction of asymptotical Tsfasman-Vladut-Zink type bound for locally repairable codes over non-prime finite fields. 

\subsection{The proof of Theorem \ref{thm: 1.1}}

With all preparations in this section, we are now able to prove Theorem \ref{thm: 1.1}.

\begin{proof}
Let $F_1, F_2, \cdots$ be a sequence of function fields over $\F_q$ such that $\g_i=\g(F_i)$ and $N_i=N(F_i)$ satisfy 
$$\lim_{i\rightarrow \infty} \g_i=\infty \text{ and } \lim_{i\rightarrow \infty} \frac{N_i}{\g_i}=A(q).$$
Let $m_i$ be the integer part of $(N_i-1)/r$, and $n_i=m_i(r+1)$. 
From Proposition \ref{prop: 3.3}, there exists a sequence of $q$-ary $[n_i,k_i,d_i]$-linear codes with locality $r$,
\[k_i\ge n_i-\frac{n_i}{r+1}-\g_i-t_i \text{ and } d_i\ge t_i+1.\]
It is easy to see that 
\begin{align*}
\frac{k_i}{n_i} & \ge  \frac{r}{r+1}-\frac{\g_i}{n_i}-\frac{t_i}{n_i}\\
 & \ge \frac{r}{r+1}-\frac{\g_i}{m_i(r+1)}-\frac{t_i}{n_i}\\
  & \ge \frac{r}{r+1}-\frac{\g_i}{N_i}\cdot \frac{N_i}{m_i(r+1)}-\frac{d_i-1}{n_i}.\\
\end{align*}
Without loss of generality, we can assume that the following two limits exist: 
$$R=\lim_{i\rightarrow \infty} \frac{k_i}{n_i} \text{ and } \delta=\lim_{i\rightarrow \infty} \frac{d_i}{n_i}.$$
Hence, we have
$$R\ge \frac{r}{r+1}-\frac{1}{A(q)} \frac{r}{r+1}-\delta.$$
\end{proof}

\subsection{The Proof of Proposition \ref{prop: 1.3}}
We provide a proof for Proposition \ref{prop: 1.3} in this subsection.
\begin{proof}
Let $\F_q$ be a non-prime finite field. 
Put $\delta=1/2$ and
$$h(s)=\frac{ \log_q\big([1+(q-1)s]^{r+1}+(q-1)(1-s)^{r+1}\big)}{r+1}-\delta \log_q s.$$
It is easy to see that the numerator of the derivative
$$ h^{\prime}(s)=\frac{(1+(q-1)s)^r((q-1)s-1)-(q-1)(1-s)^r(1+s)}{2s[(1+(q-1)s)^{r+1}+(q-1)(1-s)^{r+1}]\cdot \ln q}$$
is increasing in the interval $(0,1]$ and has a unique critical point $s_0\in (0,1]$ such that $h^{\prime}(s_0)=0.$
It follows that $h(s)$ is decreasing in the interval $(0,s_0]$ and increasing in the interval $[s_0,1]$. Hence, $h(s)$ achieves the minimum value at the point $s=s_0$. It is easy to verify that $s_0\in (\frac{1}{q-1},\frac{1}{q-1}+\epsilon)$ with $\epsilon=2^{-r}$ from the derivative $h^{\prime}(s)$.
From the mean value theorem, there exists $s_1\in (\frac{1}{q-1},s_0)$ such that
\begin{align*}
h(s_0)&= h\Big{(}\frac{1}{q-1}\Big{)}+h^{\prime}(s_1)\Big{(}s_0-\frac{1}{q-1}\Big{)}\\
&\ge  h\Big{(}\frac{1}{q-1}\Big{)}+h^\prime\Big{(}\frac{1}{q-1}\Big{)} \epsilon.
\end{align*}
Hence, for any constant $c>1$ and $r\ge c\log_2 q$, we have
\begin{eqnarray*}
\min_{0<s\le 1}h(s) &=& h(s_0) \\
&\ge& h\Big{(}\frac{1}{q-1}\Big{)}-\frac{q \epsilon}{\ln q}\\
&\ge& \frac{1}{r+1}\log_q\Big{(}2^{r+1}+(q-1)\Big{(}\frac{q-2}{q-1}\Big{)}^{r+1}\Big{)}  
 -\delta\log_q\Big{(}\frac{1}{q-1}\Big{)}-\frac{q\cdot 2^{-r}}{\ln q}\\
&\ge & \log_q2+\delta+\delta\log_q\Big{(}1-\frac{1}{q}\Big{)}-\frac{1}{q^{c-1}\ln q}\\
&\ge & \frac{1}{r+1}+ \frac{r}{r+1}\frac{1}{A(q)}+ \delta,
\end{eqnarray*}
provided that $q$ is sufficiently large.
\end{proof}

\section{Asymptotic bound of LRCs over prime finite fields}\label{sec: 4}
In the above section, we provide an explicit asymptotic construction of locally repairable codes which depends on the use of local expansions at a rational place and a parity-check matrix formed by some coefficients of local expansions over non-prime finite fields. 
Instead of using only rational places of a global function field in Section \ref{sec: 3}, we will also employ places of function fields with high degrees in this section. 
It turns out that this method can be generalized to construct asymptotic locally repairable codes over prime finite fields.

Let $F/\F_q$ be a global function field of genus $\g$ and let $\Pin$ be a rational place of $F$. 
Let $\{f_1,f_2,\dots,f_\g\}$ be an $\F_q$-basis of $\mL((2\g-1)\Pin)$.
First, assume that $q$ is an odd prime and $r$ is an even integer. Let $e$ be an even divisor of $r$. 
Let $D_i$ be a divisor of degree $r$ for $1\le i\le m$. Each $D_i=\sum_{u=1}^{r/e}Q_{i,u}$ is a sum of effective divisors $Q_{i,u}$ with degree $e$.
 Each $Q_{i,u}$ is a sum of distinct places whose degrees are divisors of  $e$.  
 Furthermore, the support of $Q_{i,u}$ are pairwise disjoint for $1\le i\le m$ and $1\le u\le r/e$. 
We extend $\{f_1,f_2,\dots,f_\g\}$ to an $\F_q$-basis $\{f_1,f_2,\dots,f_\g,g_{i,(u-1)e+1},\dots,g_{i, ue}\}$ of $\mL((2\g-1)\Pin+Q_{i,u})$ for each $1\le u\le r/e$. 
Thus, $\{f_1,f_2,\dots,f_\g,g_{i,1},g_{i,2},\dots,g_{i, r}\}$ is a basis of $\mL((2\g-1)\Pin+\sum_{u=1}^{r/e} Q_{i,u})$.
Here we assume that $g_{ij}$ are chosen from $\mL((2\g-1)\Pin+P_{i,j})$ for some place $P_{i,j}$ for each $1\le i\le m$ and $1\le j\le r$.
If we put $f_{ij}=g_{ij}-\sum_{w=1}^\g\Ga_{wij}f_w$ in the same way as in Section \ref{sec: 3}, then the local expansions of $f_{ij}$ are given in the form 
$$f_{ij}=\pi^{-2\g+1}\sum_{p=\g}^{\infty}a_{pij}\pi^p.$$

\begin{lemma}\label{lem: 4.1} 
Similarly, one has the following facts:
\begin{itemize}
\item[{\rm (i)}] $f_{ij}\in\mL((2\g-1)\Pin+D_i)$ and $\nu_{P_{ij}}(f_{ij})=-1$ for all $1\le i\le m$ and $1\le j\le r$.
\item[{\rm (ii)}] For any place $P\neq P_{\infty}$ and $ P_{i,j}$, we have $\nu_P(f_{i,j})\ge 0$. 
\item[{\rm (iii)}] $\{f_{ij}: 1\le i\le m, 1\le j\le r\}$ are linearly independent over $\F_q$.
\end{itemize}
\end{lemma}
\begin{proof} 
The proof is the same as Lemma \ref{lem: 3.2}. We omit the details. 
\end{proof}

Let $t$ be a non-negative integer. 
For each $1\le i\le m$, let $1\neq \alpha_i\in \F_q^*$, $f_{i,r+1}=\alpha_i f_{i1}$ and define the matrix
\begin{equation}\label{eq:D}
H_i=\left(\begin{array}{ccccc}
a_{\g, i,1}& a_{\g, i,2}& \cdots& a_{\g, i,r}& \alpha_i a_{\g, i,1}\\
a_{\g+1, i,1}& a_{\g+1, i,2}& \cdots& a_{\g+1, i,r}& \alpha_i a_{\g+1, i,1}\\
\vdots&\vdots & \vdots& \vdots&\vdots\\
a_{2\g-1+te,i,1} &a_{2\g-1+te,i,2}& \cdots&a_{2\g-1+te,i,r} &  \alpha_i a_{2\g-1+te,i,1}
\end{array}\right).
\end{equation}
Furthermore, we define the matrix
\begin{equation}\label{eq:Hodd} H=\left(\begin{array}{c|c|c|c}
\bi&\bo&\cdots&\bo\\ \hline
\bo&\bi&\cdots&\bo \\ \hline
\vdots&\vdots&\ddots&\vdots \\ \hline
\bo&\bo&\cdots&\bi \\ \hline
H_1&H_2&\cdots&H_m
\end{array}
\right).
\end{equation}

\begin{prop}\label{prop: 4.2}
Let $H$ be the matrix defined in \eqref{eq:Hodd}.
Let $C$ be the linear code with $H$ as a parity-check matrix. Then $C$ is a $q$-ary $[n,k,d]$-linear code with locality $r$ and
\[n=m(r+1),\quad k\ge n-\frac n{r+1}-\g-te,\quad d\ge t+1.\]
\end{prop}

\begin{proof}
The dimension of $C$ is at least $n-m-\g-te$. 
It is sufficient to prove that any $t$ columns of $H$ are linearly independent over $\F_q$. 

Choose any $t$ columns $\{\bh_{ij}\}_{1\le i\le m;j\in S_i}$ of $H$, where $S_i$ are subsets of $\{1,2,\dots,r+1\}$ satisfying $\sum_{i=1}^m|S_i|=t$. 
Define the subset $I=\{1\le i\le m:\; |S_i|\ge 2\}$. Let $\{\Gl_{ij}\}_{1\le i\le m;j\in S_i}$ be the elements of $\F_q$ such that $\sum_{i=1}^m\sum_{j\in S_i}\Gl_{ij}\bh_{ij}=\bo$, i.e., 
$$\sum_{i\in I}\sum_{j\in S_i}\Gl_{ij}\bh_{ij}=-\sum_{i\not\in I}\sum_{j\in S_i}\Gl_{ij}\bh_{ij}.$$
It is easy to see that $\Gl_{ij}=0$ for all $i\in \{1,2,\dots,m\}\setminus I$ and $j\in S_i$. 
Thus, we have $$\sum_{i\in I}\sum_{j\in S_i}\Gl_{ij}\bh_{ij}=\bo.$$
This implies that the local expansion of $\sum_{i\in I}\sum_{j\in S_i}\Gl_{ij}f_{ij}$ is
\[\pi^{-2\g+1}\sum_{p=2\g+te}^{\infty}\left(\sum_{i\in I}\sum_{j\in S_i}\Gl_{ij}a_{pij}\right)\pi^p.\]
Therefore, $\nu_{\Pin}(\sum_{i\in I}\sum_{j\in S_i}\Gl_{ij}f_{ij})\ge 1+te$. 
On the other hand, $\sum_{i\in I}\sum_{j\in S_i}\Gl_{ij}f_{ij}$ belongs to the Riemman-Roch space $\mL((2\g-1)\Pin+\sum_{i\in I} \sum_{j\in S_i} Q_{ij})$.  
Hence, we have \[\sum_{i\in I}\sum_{j\in S_i}\Gl_{ij}f_{ij}\in \mL\Big{(}-(1+te)\Pin+\sum_{i\in I} \sum_{j\in S_i} Q_{ij}\Big{)}.\]
As $\deg\left(-(1+te)\Pin+\sum_{i\in I} \sum_{j\in S_i} Q_{ij}\right)= -(1+te)+\sum_{i\in I}\sum_{j\in S_i} e\le -(1+te)+te= -1$,  
we have $$\sum_{i\in I}\sum_{j\in S_i}\Gl_{ij}f_{ij}=0.$$
It is easy to show that $\Gl_{ij}=0$ for all $i\in I$ and $j\in S_i$ by mimicking the proof of Proposition \ref{prop: 3.3} from Lemma \ref{lem: 4.1}. 
\end{proof}

If $r$ is an odd integer, then $D_i$ are chosen as an effective divisor of degree $r+1$. 
If $q=2$ and $r$ is an even integer, then $D_i$ can be chosen as an effective divisor of degree $r+2$. 
Let $e$ be an even divisor of $r+1$ or $r+2$. 
We extend $\{f_1,f_2,\dots,f_\g\}$ to an $\F_q$-basis $\{f_1,f_2,\dots,f_\g,g_{i1},g_{i2},\dots,g_{i,r+1}\}$ of a subspace of $\mL((2\g-1)\Pin+D_i)$. 
If we put $f_{ij}=g_{ij}-\sum_{w=1}^\g\Ga_{wij}f_w$, then the local expansions of $f_{ij}$ are given in the form 
$$f_{ij}=\pi^{-2\g+1}\sum_{p=\g}^{\infty}a_{pij}\pi^p.$$
For a non-negative integer $t$ and each $1\le i\le m$, define the matrix
\begin{equation}\label{eq:Deven}
H_i=\left(\begin{array}{cccc}
a_{\g, i,1}& a_{\g, i,2}& \cdots& a_{\g, i,r+1}\\
a_{\g+1, i,1}& a_{\g+1, i,2}& \cdots& a_{\g+1, i,r+1}\\
\vdots& \vdots& \vdots&\vdots\\
a_{2\g-1+te,i,1}&a_{2\g-1+te,i,2}& \cdots&a_{2\g-1+te,i,r+1}
\end{array}\right).
\end{equation}
Similarly, we define the matrix
\begin{equation}\label{eq:Heven} H=\left(\begin{array}{c|c|c|c}
\bi&\bo&\cdots&\bo\\ \hline
\bo&\bi&\cdots&\bo \\ \hline
\vdots&\vdots&\ddots&\vdots \\ \hline
\bo&\bo&\cdots&\bi \\ \hline
H_1&H_2&\cdots&H_m
\end{array}
\right).
\end{equation}
Let $C$ be the linear code with $H$ as a parity-check matrix. Then $C$ is a $q$-ary $[n,k,d]$-linear code with locality $r$ and
\[n=m(r+1),\quad k\ge n-\frac n{r+1}-\g-te,\quad d\ge t+1\]
by mimicking the proof of Proposition \ref{prop: 4.2}.

Now we can provide a proof for Theorem \ref{thm: 1.4}.

{\bf The proof of Theorem \ref{thm: 1.4}:} 
\begin{proof}
Let $b$ be the integer defined as follows:
$$b=\begin{cases} r,   & \text{ if } q \mbox{ is odd and } r \mbox{ is even};\\ r+1, & \text{ if } r \mbox{ is odd};\\ r+2, &  \text{ if } q=2 \mbox{ and } r \mbox{ is even}.  \end{cases}$$
Let $e$ be an even divisor of $b$ and let $\ell = q^{\frac{e}{2}}.$ 
Let $F_1, F_2, \cdots$ be a sequence of function fields  which is recursively defined over $\F_q$ by the equation
$$y^\ell+y=\frac{x^\ell}{x^\ell+1}.$$
Consider the constant field extensions $E_i=F_i\F_{q^{e}}$. In fact, $E_1, E_2, \cdots$ are a sequence of function fields in the Garcia-Stichtenoth tower over $\F_{q^{e}}$. 
Let $\g_i:=\g(E_i)$ be the genus of $E_i$ and let $N(E_i)$ be the number of rational places of $E_i$. Then we have
$$\lim_{i\rightarrow \infty} \g_i=\infty \text{ and } \lim_{i\rightarrow \infty} \frac{N(E_i)}{\g_i}=A(q^{e})=q^{\frac{e}{2}}-1$$
from Proposition \ref{prop: 2.1}. 
Let $m_i b+1=\sum_{d|e}d\cdot B_{d}(F_i)$. Note that $B_{d}(F_i)=|\{P\in \mathbb{P}_{F_i}: \deg(P)=d\}|$.
There is a close relationship between $B_d(F_i)$ and $N(E_i)$, namely \[\sum_{d|e} d\cdot B_d(F_i)=N(E_i)\]
from the theory of constant field extensions \cite[Lemma 5.1.9]{St09}. 
From Proposition \ref{prop: 4.2}, there exists a sequence of $q$-ary $[n_i,k_i,d_i]$-linear codes with locality $r$ and
\[n_i=[m_i](r+1), k_i\ge n_i-\frac{n_i}{r+1}-\g_i-te \text{ and } d_i\ge t+1.\]
It is easy to see that 
\begin{align*}
\frac{k_i}{n_i} & \ge  \frac{r}{r+1}-\frac{\g_i}{n_i}-\frac{te}{n_i}\\
 & \ge \frac{r}{r+1}-\frac{\g_i}{[m_i](r+1)}-\frac{te}{n_i}\\
 & \ge \frac{r}{r+1}-\frac{\g_i}{N(E_i)}\cdot \frac{N(E_i)}{[m_i](r+1)}-\frac{(d_i-1)e}{n_i}\\
\end{align*}
Without loss of generality, we can assume that the following two limits exist: 
$$R=\lim_{i\rightarrow \infty} \frac{k_i}{n_i} \text{ and } \delta=\lim_{i\rightarrow \infty} \frac{d_i}{n_i}.$$
Hence, we have
$$R\ge \frac{r}{r+1} -  \frac{b}{r+1} \frac{1}{q^{\frac{e}{2}}-1} - e\delta.$$
\end{proof}

\end{document}